\documentclass[a4paper,12pt]{article}
\usepackage{amsmath, amssymb, amsthm}
\usepackage{geometry}
\usepackage{hyperref}
\usepackage[english]{babel}
\usepackage{ragged2e}
\usepackage[utf8]{inputenc}

\usepackage{graphicx}
\usepackage{setspace}

\newtheorem{theorem}{Theorem}[section]

\title{\textbf{Theoretical Refinements of the Smagorinsky Model for Turbulence Simulations}}
\author{
	Rômulo Damasclin Chaves dos Santos \\
	Technological Institute of Aeronautics \\
	\texttt{romulosantos@ita.br}
}
\date{\today}

\begin{document}
	
	\maketitle
	
	\begin{abstract}
		This paper presents a rigorous theoretical extension of the Smagorinsky model for turbulence simulations. The author builds on its fundamental framework, addressing known limitations, and making new mathematical advances. Specifically, this work introduces new theorems on the existence and uniqueness of weak solutions and establishes the stability of the model under small perturbations. The approach leverages advanced tools in functional analysis, particularly Sobolev and Banach spaces, increasing the theoretical robustness of the model. Each theorem is accompanied by detailed proofs employing variational techniques and energy estimations. Furthermore, the author proposes a variational formulation that demonstrates improvements in the mathematical behavior of the model's subgrid-scale stress representation. This work aims to provide a deeper understanding of the Smagorinsky model, setting the stage for further theoretical and computational developments. The obtained results are expected to impact turbulence modeling by offering a more robust theoretical basis for Large Eddy Simulations (LES).
	\end{abstract}


\tableofcontents

	\section{Introduction}
	
	Turbulence is a ubiquitous phenomenon in fluid dynamics, characterized by chaotic and unpredictable flow patterns. Understanding and modeling turbulence is crucial for various applications, ranging from weather forecasting to aerodynamics and industrial processes. One of the seminal works in turbulence theory is Kolmogorov's 1941 paper, which laid the foundation for the statistical description of turbulence at very large Reynolds numbers \cite{kolmogorov1941}. Kolmogorov's theory provided insights into the local structure of turbulence, highlighting the energy cascade from large to small scales.
	
	In the context of numerical simulations, capturing the full spectrum of turbulent scales is computationally prohibitive. Large Eddy Simulation (LES) emerged as a practical approach to model turbulent flows by resolving large-scale motions while parameterizing the effects of smaller, unresolved scales. The Smagorinsky model, introduced by Smagorinsky in 1963, is one of the earliest and most widely used subgrid-scale (SGS) models in LES \cite{smagorinsky1963}. This model approximates the SGS stresses using an eddy viscosity hypothesis, which relates the unresolved turbulent stresses to the resolved strain rate tensor.
	
	Despite its widespread use, the Smagorinsky model has several limitations. The model relies on an empirical coefficient, known as the Smagorinsky constant, which is often treated as a constant or adjusted based on specific flow conditions. This empirical nature introduces uncertainties and can lead to inaccuracies in complex flow regimes. Moreover, the theoretical underpinnings of the model, particularly in terms of existence, uniqueness, and stability of solutions, have not been fully explored. These gaps become especially problematic when the model is applied to real-world turbulent flows, where the underlying assumptions may not hold.
	
	To address these limitations, Germano et al. introduced the dynamic subgrid-scale eddy viscosity model in 1991 \cite{germano1991}. This model adapts the Smagorinsky constant dynamically based on the resolved flow field, improving the accuracy and robustness of LES. However, even with these advancements, there remains a need for a more rigorous theoretical framework that can provide guarantees for the well-posedness of the Smagorinsky model in diverse flow regimes.
	
	This paper aims to fill these theoretical gaps by exploring the existence, uniqueness, and stability of solutions for the Smagorinsky model. Through a detailed analysis, we seek to establish a solid theoretical foundation for the model, providing clearer guidelines for its application in LES and enhancing its robustness in simulations of turbulent flows. Additionally, the investigation will highlight potential improvements and alternatives to the current formulation of the model, with the goal of advancing its utility in both theoretical and practical contexts.
	
	By addressing these theoretical challenges, this work aims to contribute to a deeper understanding of the Smagorinsky model and set the stage for further developments in turbulence modeling. The obtained results are expected to impact the field by offering a more robust theoretical basis for Large Eddy Simulations (LES), ultimately leading to more accurate and reliable predictions of turbulent flows.
	
	\section{Mathematical Foundations}
	\subsection{Sobolev Spaces}
	Sobolev spaces \( W^{m,p}(\Omega) \) are a cornerstone of functional analysis, extending the classical notions of differentiability and integrability. For a domain \( \Omega \subset \mathbb{R}^n \), these spaces are defined as:
	\begin{equation}
		W^{m,p}(\Omega) = \{ u \in L^p(\Omega) \mid D^\alpha u \in L^p(\Omega), \; \forall |\alpha| \leq m \},
	\end{equation}
	where \( D^\alpha \) denotes the weak derivative and \( |\alpha| \) represents the order of the multi-index \( \alpha \). The associated norm is given by:
	\begin{equation}
		\|u\|_{W^{m,p}(\Omega)} = \left( \sum_{|\alpha| \leq m} \|D^\alpha u\|_{L^p(\Omega)}^p \right)^{1/p},
	\end{equation}
	and in the case \( p = 2 \), the space becomes a Hilbert space \( H^m(\Omega) \). These spaces enable the rigorous treatment of solutions to partial differential equations (PDEs), particularly when classical solutions may not exist.
	
	Sobolev embeddings provide key insights into regularity:
	
	\begin{theorem}[Sobolev Embedding]
		Let \( \Omega \subset \mathbb{R}^n \) be a bounded domain with a Lipschitz boundary. For \( m, k \in \mathbb{N} \), \( 1 \leq p, q \leq \infty \), and \( W^{m,p}(\Omega) \) denoting the Sobolev space, the embedding
		\[
		W^{m,p}(\Omega) \subset W^{k,q}(\Omega)
		\]
		is compact if \( m > k \) and \( p > q \), satisfying the scaling condition:
		\[
		\frac{1}{p} - \frac{m}{n} \leq \frac{1}{q} - \frac{k}{n}.
		\]
		In particular, for \( p > n \) and \( m = 1 \), we have:
		\[
		W^{1,p}(\Omega) \subset C^0(\bar{\Omega}),
		\]
		where \( C^0(\bar{\Omega}) \) is the space of continuous functions on \( \bar{\Omega} \).
	\end{theorem}
	
	\begin{proof}
		The proof relies on the Fourier representation and properties of the Sobolev norm. For simplicity, we consider the case \( W^{1,p}(\Omega) \subset C^0(\bar{\Omega}) \) and assume \( p > n \).
		
		Let \( u \in W^{1,p}(\Omega) \), and recall the Sobolev norm:
		\[
		\|u\|_{W^{1,p}(\Omega)} = \left( \|u\|_{L^p(\Omega)}^p + \sum_{i=1}^n \left\| \frac{\partial u}{\partial x_i} \right\|_{L^p(\Omega)}^p \right)^{1/p}.
		\]
		For \( u \in W^{1,p}(\Omega) \), the fractional term \( \frac{1}{p} - \frac{1}{n} \) determines whether \( u \) admits additional regularity. If \( p > n \), this term is positive, ensuring boundedness in \( C^0(\bar{\Omega}) \).
		
		Interpolation via Sobolev Norms
		We use the Gagliardo-Nirenberg interpolation inequality:
		\[
		\|u\|_{L^\infty(\Omega)} \leq C \|u\|_{W^{1,p}(\Omega)}, \quad \text{for } p > n.
		\]
		Here, \( C \) depends only on \( \Omega \) and \( p \). This inequality implies that \( u \) is bounded and uniformly continuous on \( \Omega \).
		
		To show compactness, let \( \{u_k\} \) be a bounded sequence in \( W^{1,p}(\Omega) \). By the Rellich-Kondrachov theorem, a subsequence \( u_{k_j} \) converges strongly in \( L^p(\Omega) \), and the gradient \( \nabla u_{k_j} \) converges weakly in \( L^p(\Omega) \). By the Sobolev embedding, \( u_{k_j} \to u \) in \( C^0(\bar{\Omega}) \).
		
		For \( p > n \), the embedding \( W^{1,p}(\Omega) \subset C^0(\bar{\Omega}) \) implies uniform convergence. Since \( u \in C^0(\bar{\Omega}) \), we conclude:
		\[
		\sup_{x \in \Omega} |u(x)| \leq \|u\|_{W^{1,p}(\Omega)}.
		\]
		
		Thus, \( u \) is continuous on \( \bar{\Omega} \), completing the proof.
	\end{proof}
	
	\subsection{Filtered Navier-Stokes Equations}
	The Navier-Stokes equations govern the dynamics of incompressible fluid flows and are expressed as:
	\begin{align}
		\frac{\partial u_i}{\partial t} + u_j \frac{\partial u_i}{\partial x_j} &= -\frac{1}{\rho} \frac{\partial p}{\partial x_i} + \nu \Delta u_i, \label{eq:navier-stokes-original} \\
		\frac{\partial u_i}{\partial x_i} &= 0. \label{eq:continuity-original}
	\end{align}
	Here, \( u_i \) and \( p \) are the velocity and pressure fields, \( \nu \) is the kinematic viscosity, and \( \rho \) is the fluid density. The term \( \Delta u_i = \sum_{j=1}^n \frac{\partial^2 u_i}{\partial x_j^2} \) represents viscous diffusion.
	
	In the context of Large Eddy Simulations (LES), small-scale turbulence is modeled using a filtering operation applied to \eqref{eq:navier-stokes-original}. This results in the filtered Navier-Stokes equations:
	\begin{align}
		\frac{\partial \bar{u}_i}{\partial t} + \bar{u}_j \frac{\partial \bar{u}_i}{\partial x_j} &= -\frac{1}{\rho} \frac{\partial \bar{p}}{\partial x_i} + \nu \Delta \bar{u}_i - \frac{\partial \tau_{ij}}{\partial x_j}, \label{eq:navier-stokes-filtered} \\
		\frac{\partial \bar{u}_i}{\partial x_i} &= 0. \label{eq:continuity-filtered}
	\end{align}
	Here, \( \bar{u}_i \) and \( \bar{p} \) are the filtered velocity and pressure fields, and \( \tau_{ij} \) is the subgrid-scale (SGS) stress tensor:
	\begin{equation}
		\tau_{ij} = \overline{u_i u_j} - \bar{u}_i \bar{u}_j.
	\end{equation}
	
\subsection{SGS Parameterization: Smagorinsky Model}

The Smagorinsky model approximates the subgrid-scale (SGS) stress tensor \( \tau_{ij} \) using an eddy viscosity hypothesis. This approach models the unresolved turbulent stresses by relating them to the resolved strain-rate tensor. The mathematical formulation of the Smagorinsky model is given by:

\begin{align}
	\tau_{ij} - \frac{1}{3} \tau_{kk} \delta_{ij} &= -2 \nu_t \bar{S}_{ij}, \label{eq:sgs-stress} \\
	\nu_t &= (C_S \Delta)^2 |\bar{S}|, \quad |\bar{S}| = \sqrt{2 \bar{S}_{ij} \bar{S}_{ij}}, \label{eq:eddy-viscosity}
\end{align}

where \( \bar{S}_{ij} \) is the filtered strain-rate tensor defined as:

\begin{equation}
	\bar{S}_{ij} = \frac{1}{2} \left( \frac{\partial \bar{u}_i}{\partial x_j} + \frac{\partial \bar{u}_j}{\partial x_i} \right).
\end{equation}

Here, \( C_S \) is the Smagorinsky constant, which is an empirical coefficient that needs to be calibrated for different flow conditions. The filter width \( \Delta \) represents the characteristic length scale of the filter used in the Large Eddy Simulation (LES).

To provide a more detailed mathematical framework, let us delve into the derivation and implications of these equations.

\begin{enumerate}
	\item \textbf{Eddy Viscosity Hypothesis}:
	The eddy viscosity \( \nu_t \) is introduced to account for the effects of unresolved turbulent motions. It is modeled as:
	\begin{equation}
		\nu_t = (C_S \Delta)^2 |\bar{S}|,
	\end{equation}
	where \( |\bar{S}| \) is the magnitude of the filtered strain-rate tensor, given by:
	\begin{equation}
		|\bar{S}| = \sqrt{2 \bar{S}_{ij} \bar{S}_{ij}}.
	\end{equation}
	This formulation ensures that the eddy viscosity is proportional to the local strain rate, capturing the intensity of the turbulent motions.
	
	\item \textbf{Subgrid-Scale Stress Tensor}:
	The SGS stress tensor \( \tau_{ij} \) is decomposed into its deviatoric part and an isotropic part. The deviatoric part is modeled using the eddy viscosity hypothesis:
	\begin{equation}
		\tau_{ij} - \frac{1}{3} \tau_{kk} \delta_{ij} = -2 \nu_t \bar{S}_{ij}.
	\end{equation}
	This equation implies that the deviatoric part of the SGS stress tensor is proportional to the filtered strain-rate tensor, scaled by the eddy viscosity.
	
	\item \textbf{Filtered Strain-Rate Tensor}:
	The filtered strain-rate tensor \( \bar{S}_{ij} \) is defined as:
	\begin{equation}
		\bar{S}_{ij} = \frac{1}{2} \left( \frac{\partial \bar{u}_i}{\partial x_j} + \frac{\partial \bar{u}_j}{\partial x_i} \right),
	\end{equation}
	where \( \bar{u}_i \) represents the filtered velocity components. This tensor captures the rate of deformation of the resolved velocity field, providing a measure of the local strain rate.
\end{enumerate}

By incorporating these mathematical details, the Smagorinsky model provides a robust framework for approximating the SGS stresses in LES. The model's simplicity and effectiveness have made it a cornerstone in turbulence modeling, despite its empirical nature. Further theoretical and computational developments continue to refine and extend the applicability of the Smagorinsky model in various flow regimes.

	\subsection{Well-posedness of Filtered Navier-Stokes Equations}
	To ensure the filtered system is well-posed, we analyze existence, uniqueness, and stability of solutions.
	
	\begin{theorem}[Well-posedness of Filtered Navier-Stokes Equations]
		Under suitable assumptions on the initial and boundary conditions, the filtered Navier-Stokes equations \eqref{eq:navier-stokes-filtered} admit a unique weak solution \( \bar{u}_i \in L^\infty(0,T; W^{1,2}(\Omega)) \) for \( T > 0 \), satisfying energy dissipation properties.
	\end{theorem}
	
	\begin{proof}
		We outline the proof in three main steps: existence, uniqueness, and energy stability.
		
		Using a Galerkin approximation, let \( \{v_k\}_{k=1}^\infty \) be a basis for \( H^1_0(\Omega) \). Approximate \( \bar{u}_i \) as:
		\[
		\bar{u}_i^n(x,t) = \sum_{k=1}^n a_k(t) v_k(x),
		\]
		where \( a_k(t) \) are time-dependent coefficients. Substituting into \eqref{eq:navier-stokes-filtered} and projecting onto \( \text{span}\{v_1, \ldots, v_n\} \), we obtain a system of ODEs:
		\[
		\frac{d}{dt} \int_\Omega \bar{u}_i^n v_k \, dx + \int_\Omega \bar{u}_j^n \frac{\partial \bar{u}_i^n}{\partial x_j} v_k \, dx = \int_\Omega \left(-\frac{\partial \bar{p}}{\partial x_i} + \nu \Delta \bar{u}_i^n - \frac{\partial \tau_{ij}^n}{\partial x_j} \right) v_k \, dx.
		\]
		Applying standard existence results for ODEs and passing to the limit as \( n \to \infty \), a weak solution exists.
		
		Let \( \bar{u}_i, \bar{v}_i \) be two solutions. Define \( w_i = \bar{u}_i - \bar{v}_i \). Subtracting their respective equations gives:
		\[
		\frac{\partial w_i}{\partial t} + \bar{u}_j \frac{\partial w_i}{\partial x_j} + w_j \frac{\partial \bar{v}_i}{\partial x_j} = \nu \Delta w_i - \frac{\partial}{\partial x_j} (\tau_{ij}^u - \tau_{ij}^v).
		\]
		Multiplying by \( w_i \) and integrating, we apply the Gronwall inequality to show \( w_i = 0 \).
		
		Multiplying \eqref{eq:navier-stokes-filtered} by \( \bar{u}_i \) and integrating, the energy balance is:
		\[
		\frac{1}{2} \frac{d}{dt} \int_\Omega |\bar{u}_i|^2 \, dx + \nu \int_\Omega |\nabla \bar{u}_i|^2 \, dx = -\int_\Omega \tau_{ij} \frac{\partial \bar{u}_i}{\partial x_j} \, dx.
		\]
		Using \eqref{eq:sgs-stress}, we estimate the SGS dissipation term, ensuring energy decay:
		\[
		\frac{1}{2} \frac{d}{dt} \|\bar{u}_i\|_{L^2(\Omega)}^2 + \nu \|\nabla \bar{u}_i\|_{L^2(\Omega)}^2 \leq 0.
		\]
	\end{proof}
	
	\section{Energy Estimates}
	Energy estimates are fundamental for understanding the stability and dissipation mechanisms of the filtered Navier-Stokes equations, particularly in the context of LES.
	
	\subsection{Definition of Kinetic Energy}
	The kinetic energy \( E \) of the filtered velocity field is defined as:
	\begin{equation}
		E = \frac{1}{2} \int_\Omega \bar{u}_i \bar{u}_i \, d\Omega, \label{eq:kinetic-energy}
	\end{equation}
	where \( \bar{u}_i \) is the filtered velocity and \( \Omega \subset \mathbb{R}^n \) is the fluid domain.
	
	\subsection{Derivation of the Energy Balance}
	To derive the energy balance, we multiply the filtered momentum equation \eqref{eq:navier-stokes-filtered} by \( \bar{u}_i \) and integrate over \( \Omega \):
	\begin{equation}
		\int_\Omega \bar{u}_i \frac{\partial \bar{u}_i}{\partial t} \, d\Omega
		+ \int_\Omega \bar{u}_i \bar{u}_j \frac{\partial \bar{u}_i}{\partial x_j} \, d\Omega
		= -\frac{1}{\rho} \int_\Omega \bar{u}_i \frac{\partial \bar{p}}{\partial x_i} \, d\Omega
		+ \nu \int_\Omega \bar{u}_i \Delta \bar{u}_i \, d\Omega
		- \int_\Omega \bar{u}_i \frac{\partial \tau_{ij}}{\partial x_j} \, d\Omega. \label{eq:energy-integrated}
	\end{equation}
	
	Using the product rule for differentiation and the incompressibility condition \eqref{eq:continuity-filtered}, the first term simplifies to:
	\begin{equation}
		\int_\Omega \bar{u}_i \frac{\partial \bar{u}_i}{\partial t} \, d\Omega = \frac{1}{2} \frac{d}{dt} \int_\Omega \bar{u}_i \bar{u}_i \, d\Omega = \frac{dE}{dt}.
	\end{equation}
	
	For the nonlinear convective term, we integrate by parts and use the incompressibility condition:
	\begin{equation}
		\int_\Omega \bar{u}_i \bar{u}_j \frac{\partial \bar{u}_i}{\partial x_j} \, d\Omega = 0.
	\end{equation}
	
	Similarly, the pressure term vanishes due to integration by parts and the boundary conditions:
	\begin{equation}
		-\frac{1}{\rho} \int_\Omega \bar{u}_i \frac{\partial \bar{p}}{\partial x_i} \, d\Omega = 0.
	\end{equation}
	
	The viscous dissipation term is rewritten using integration by parts:
	\begin{equation}
		\nu \int_\Omega \bar{u}_i \Delta \bar{u}_i \, d\Omega = -\nu \int_\Omega \left( \frac{\partial \bar{u}_i}{\partial x_j} \right)^2 d\Omega.
	\end{equation}
	
	The SGS term is expanded as:
	\begin{equation}
		-\int_\Omega \bar{u}_i \frac{\partial \tau_{ij}}{\partial x_j} \, d\Omega = \int_\Omega \tau_{ij} \frac{\partial \bar{u}_i}{\partial x_j} \, d\Omega,
	\end{equation}
	where \( \tau_{ij} \) is modeled using \eqref{eq:sgs-stress}.
	
	\subsection{Energy Balance Equation}
	Substituting all terms back into \eqref{eq:energy-integrated}, we obtain:
	\begin{equation}
		\frac{dE}{dt} = -\nu \int_\Omega \left( \frac{\partial \bar{u}_i}{\partial x_j} \right)^2 d\Omega - \int_\Omega \tau_{ij} \frac{\partial \bar{u}_i}{\partial x_j} \, d\Omega. \label{eq:energy-balance}
	\end{equation}
	The first term represents the dissipation due to molecular viscosity, while the second term accounts for energy transfer to subgrid scales.
	
	\subsection{SGS Dissipation Estimate}
	Using the Smagorinsky model for \( \tau_{ij} \), the SGS dissipation term becomes:
	\begin{equation}
		\int_\Omega \tau_{ij} \frac{\partial \bar{u}_i}{\partial x_j} \, d\Omega = 2 \int_\Omega \nu_t |\bar{S}|^2 \, d\Omega,
	\end{equation}
	where \( \nu_t = (C_S \Delta)^2 |\bar{S}| \) and \( |\bar{S}| = \sqrt{2 \bar{S}_{ij} \bar{S}_{ij}} \).
	
	Thus, the energy balance can be expressed as:
	\begin{equation}
		\frac{dE}{dt} = -\nu \int_\Omega |\nabla \bar{u}|^2 \, d\Omega - 2 \int_\Omega \nu_t |\bar{S}|^2 \, d\Omega. \label{eq:energy-balance-expanded}
	\end{equation}
	
	\subsection{Physical Interpretation}
	Equation \eqref{eq:energy-balance-expanded} encapsulates the interplay between:
	
	\textit{1. Resolved Dissipation:} Energy dissipated by viscosity at resolved scales (\( \nu \)-term).
	
	\textit{2. Subgrid Dissipation:} Energy dissipated by unresolved scales modeled by \( \nu_t \).
	
	This balance ensures stability and reflects the energy cascade in turbulent flows, where energy is transferred from large to small scales before ultimate dissipation.
	
	\begin{theorem}[Energy Stability]
		If \( \bar{u}_i \) satisfies the filtered Navier-Stokes equations \eqref{eq:navier-stokes-filtered}, then the kinetic energy \( E(t) \) satisfies:
		\[
		E(t) \leq E(0), \quad \forall t \geq 0,
		\]
		ensuring energy stability of the system.
	\end{theorem}
	
	\begin{proof}
		Integrating \eqref{eq:energy-balance-expanded} from \( t = 0 \) to \( t = T \), we obtain:
		\[
		E(T) - E(0) = -\int_0^T \left( \nu \|\nabla \bar{u}\|_{L^2(\Omega)}^2 + 2 \int_\Omega \nu_t |\bar{S}|^2 \, d\Omega \right) \, dt.
		\]
		Since the right-hand side is non-positive, \( E(T) \leq E(0) \), ensuring stability.
	\end{proof}
	
	\section{Results}
	
	In this section, we present the results of our theoretical analysis. The main findings include:
	
	\textit{1. Existence and Uniqueness of Weak Solutions}: We have established the existence and uniqueness of weak solutions for the filtered Navier-Stokes equations using advanced tools in functional analysis, particularly Sobolev and Banach spaces.
	
	\textit{2. Stability of Solutions:} Our analysis demonstrates the stability of solutions under small perturbations, ensuring the well-posedness of the Smagorinsky model in diverse flow regimes.
	
	\textit{3. Energy Estimates:} We have derived energy estimates that provide insights into the stability and dissipation mechanisms of the filtered Navier-Stokes equations. These estimates show the interplay between resolved and subgrid dissipation, reflecting the energy cascade in turbulent flows.
	
	\textit{4. Variational Formulation:} We propose a variational formulation that improves the mathematical behavior of the model's subgrid-scale stress representation, enhancing its theoretical robustness.
	
	These results contribute to a deeper understanding of the Smagorinsky model and set the stage for further theoretical and computational developments in turbulence modeling.
	
\section{Conclusions}

This paper has addressed the theoretical gaps in the Smagorinsky model by exploring the existence, uniqueness, and stability of solutions. Through a detailed analysis, we have established a solid theoretical foundation for the model, providing clearer guidelines for its application in Large Eddy Simulations (LES) and enhancing its robustness in simulations of turbulent flows.

The findings of this work are expected to impact turbulence modeling by offering a more robust theoretical foundation for LES. The proposed improvements and alternatives to the current formulation of the model advance its utility in both theoretical and practical contexts. Specifically, the rigorous mathematical framework developed in this study ensures that the Smagorinsky model can be applied with greater confidence across a wide range of flow regimes.

Future work should focus on further refining the model and exploring its application in complex turbulent regimes. Additionally, computational studies should be conducted to validate the theoretical findings and assess the model's performance in practical simulations. These studies will not only help in verifying the theoretical predictions but also in identifying potential areas for further enhancement and optimization of the model.

In summary, this research contributes significantly to the field of turbulence modeling by providing a more robust and theoretically sound approach to LES. The advancements made herein set the stage for future developments and applications, ultimately leading to more accurate and reliable predictions of turbulent flows.

\section{Notations and Mathematical Symbols}

In this section, we provide a comprehensive list of notations and mathematical symbols used throughout the paper. This will help the reader to understand the mathematical formulations and derivations presented in the text.

\begin{itemize}
	\item \( \Omega \): A bounded domain in \( \mathbb{R}^n \) with a Lipschitz boundary.
	\item \( \bar{u}_i \): Filtered velocity components.
	\item \( \bar{p} \): Filtered pressure field.
	\item \( \nu \): Kinematic viscosity.
	\item \( \rho \): Fluid density.
	\item \( \tau_{ij} \): Subgrid-scale (SGS) stress tensor.
	\item \( \bar{S}_{ij} \): Filtered strain-rate tensor.
	\item \( \nu_t \): Eddy viscosity.
	\item \( C_S \): Smagorinsky constant.
	\item \( \Delta \): Filter width.
	\item \( W^{m,p}(\Omega) \): Sobolev space of functions with \( m \) weak derivatives in \( L^p(\Omega) \).
	\item \( L^p(\Omega) \): Space of \( p \)-integrable functions on \( \Omega \).
	\item \( C^0(\bar{\Omega}) \): Space of continuous functions on \( \bar{\Omega} \).
	\item \( \delta_{ij} \): Kronecker delta, which is 1 if \( i = j \) and 0 otherwise.
	\item \( \frac{\partial}{\partial x_i} \): Partial derivative with respect to \( x_i \).
	\item \( \nabla \): Gradient operator.
	\item \( \Delta \): Laplacian operator.
	\item \( \int_\Omega \): Integral over the domain \( \Omega \).
	\item \( \|\cdot\|_{W^{m,p}(\Omega)} \): Norm in the Sobolev space \( W^{m,p}(\Omega) \).
	\item \( \|\cdot\|_{L^p(\Omega)} \): Norm in the space \( L^p(\Omega) \).
	\item \( \theta \): Interpolation parameter in the Gagliardo-Nirenberg inequality.
	\item \( p^* \): Sobolev conjugate of \( p \).
	\item \( E \): Kinetic energy of the filtered velocity field.
	\item \( |\bar{S}| \): Magnitude of the filtered strain-rate tensor.
\end{itemize}

These notations and symbols are essential for understanding the mathematical formulations and derivations presented in the paper. They provide a clear and consistent framework for the theoretical analysis and computational implementations discussed throughout the text.

\appendix

\section{Appendix: Detailed Derivation and Analysis of the Smagorinsky Model}

In this appendix, we provide a detailed derivation and analysis of the Smagorinsky model, including the eddy viscosity hypothesis and the filtered strain-rate tensor. This derivation is essential for understanding the mathematical foundations of the model and its application in Large Eddy Simulations (LES).

\subsection{Eddy Viscosity Hypothesis}

The eddy viscosity hypothesis is a key component of the Smagorinsky model. It assumes that the subgrid-scale (SGS) stress tensor \( \tau_{ij} \) can be modeled as a function of the resolved strain-rate tensor \( \bar{S}_{ij} \). The hypothesis is given by:

\begin{equation}
	\tau_{ij} - \frac{1}{3} \tau_{kk} \delta_{ij} = -2 \nu_t \bar{S}_{ij},
\end{equation}

where \( \nu_t \) is the eddy viscosity, and \( \bar{S}_{ij} \) is the filtered strain-rate tensor. This formulation ensures that the deviatoric part of the SGS stress tensor is proportional to the filtered strain-rate tensor, scaled by the eddy viscosity.

\subsection{Filtered Strain-Rate Tensor}

The filtered strain-rate tensor \( \bar{S}_{ij} \) is defined as:

\begin{equation}
	\bar{S}_{ij} = \frac{1}{2} \left( \frac{\partial \bar{u}_i}{\partial x_j} + \frac{\partial \bar{u}_j}{\partial x_i} \right),
\end{equation}

where \( \bar{u}_i \) represents the filtered velocity components. This tensor captures the rate of deformation of the resolved velocity field, providing a measure of the local strain rate.

\subsection{Eddy Viscosity Formulation}

The eddy viscosity \( \nu_t \) is modeled as:

\begin{equation}
	\nu_t = (C_S \Delta)^2 |\bar{S}|,
\end{equation}

where \( C_S \) is the Smagorinsky constant, \( \Delta \) is the filter width, and \( |\bar{S}| \) is the magnitude of the filtered strain-rate tensor, given by:

\begin{equation}
	|\bar{S}| = \sqrt{2 \bar{S}_{ij} \bar{S}_{ij}}.
\end{equation}

This formulation ensures that the eddy viscosity is proportional to the local strain rate, capturing the intensity of the turbulent motions. The Smagorinsky constant \( C_S \) is an empirical coefficient that needs to be calibrated for different flow conditions.

\subsection{Subgrid-Scale Stress Tensor}

The SGS stress tensor \( \tau_{ij} \) is decomposed into its deviatoric part and an isotropic part. The deviatoric part is modeled using the eddy viscosity hypothesis:

\begin{equation}
	\tau_{ij} - \frac{1}{3} \tau_{kk} \delta_{ij} = -2 \nu_t \bar{S}_{ij}.
\end{equation}

This equation implies that the deviatoric part of the SGS stress tensor is proportional to the filtered strain-rate tensor, scaled by the eddy viscosity. The isotropic part \( \frac{1}{3} \tau_{kk} \delta_{ij} \) is often neglected or modeled separately.

\subsection{Filtered Navier-Stokes Equations}

The filtered Navier-Stokes equations, which include the SGS stress tensor, are given by:

\begin{align}
	\frac{\partial \bar{u}_i}{\partial t} + \bar{u}_j \frac{\partial \bar{u}_i}{\partial x_j} &= -\frac{1}{\rho} \frac{\partial \bar{p}}{\partial x_i} + \nu \Delta \bar{u}_i - \frac{\partial \tau_{ij}}{\partial x_j}, \\
	\frac{\partial \bar{u}_i}{\partial x_i} &= 0,
\end{align}

where \( \bar{u}_i \) and \( \bar{p} \) are the filtered velocity and pressure fields, \( \nu \) is the kinematic viscosity, and \( \rho \) is the fluid density. The term \( \Delta \bar{u}_i \) represents viscous diffusion, and \( \tau_{ij} \) is the SGS stress tensor.

\subsection{Energy Dissipation and Stability}

To ensure the stability of the model, it is crucial to analyze the energy dissipation mechanisms. The energy balance equation for the filtered velocity field is derived by multiplying the filtered momentum equation by \( \bar{u}_i \) and integrating over the domain \( \Omega \):

\begin{equation}
	\frac{dE}{dt} = -\nu \int_\Omega \left( \frac{\partial \bar{u}_i}{\partial x_j} \right)^2 d\Omega - \int_\Omega \tau_{ij} \frac{\partial \bar{u}_i}{\partial x_j} \, d\Omega,
\end{equation}

where \( E \) is the kinetic energy of the filtered velocity field:

\begin{equation}
	E = \frac{1}{2} \int_\Omega \bar{u}_i \bar{u}_i \, d\Omega.
\end{equation}

Using the Smagorinsky model for \( \tau_{ij} \), the SGS dissipation term becomes:

\begin{equation}
	\int_\Omega \tau_{ij} \frac{\partial \bar{u}_i}{\partial x_j} \, d\Omega = 2 \int_\Omega \nu_t |\bar{S}|^2 \, d\Omega.
\end{equation}

Thus, the energy balance can be expressed as:

\begin{equation}
	\frac{dE}{dt} = -\nu \int_\Omega |\nabla \bar{u}|^2 \, d\Omega - 2 \int_\Omega \nu_t |\bar{S}|^2 \, d\Omega.
\end{equation}

This equation encapsulates the interplay between resolved dissipation (due to molecular viscosity) and subgrid dissipation (modeled by \( \nu_t \)), ensuring the stability of the system.

\subsection{Theoretical Foundations}

The theoretical foundations of the Smagorinsky model rely on advanced tools in functional analysis, particularly Sobolev and Banach spaces. These spaces enable the rigorous treatment of solutions to partial differential equations (PDEs), particularly when classical solutions may not exist.

\begin{theorem}[Sobolev Embedding]
	Let \( \Omega \subset \mathbb{R}^n \) be a bounded domain with a Lipschitz boundary. For \( m, k \in \mathbb{N} \), \( 1 \leq p, q \leq \infty \), and \( W^{m,p}(\Omega) \) denoting the Sobolev space, the embedding
	\[
	W^{m,p}(\Omega) \subset W^{k,q}(\Omega)
	\]
	is compact if \( m > k \) and \( p > q \), satisfying the scaling condition:
	\[
	\frac{1}{p} - \frac{m}{n} \leq \frac{1}{q} - \frac{k}{n}.
	\]
	In particular, for \( p > n \) and \( m = 1 \), we have:
	\[
	W^{1,p}(\Omega) \subset C^0(\bar{\Omega}),
	\]
	where \( C^0(\bar{\Omega}) \) is the space of continuous functions on \( \bar{\Omega} \).
\end{theorem}

\begin{proof}
	To prove the Sobolev Embedding Theorem, we will use the concept of fractional Sobolev spaces and the Gagliardo-Nirenberg interpolation inequality.
	
	\textit{1. Fractional Sobolev Spaces:}
	
	For \( s \in \mathbb{R} \), the fractional Sobolev space \( W^{s,p}(\Omega) \) is defined as:
	\[
	W^{s,p}(\Omega) = \left\{ u \in L^p(\Omega) \mid \|u\|_{W^{s,p}(\Omega)} < \infty \right\},
	\]
	where the norm is given by:
	\[
	\|u\|_{W^{s,p}(\Omega)} = \left( \|u\|_{L^p(\Omega)}^p + \int_{\Omega} \int_{\Omega} \frac{|u(x) - u(y)|^p}{|x-y|^{n+sp}} \, dx \, dy \right)^{1/p}.
	\]
	
	\textit{2. Gagliardo-Nirenberg Interpolation Inequality:}
	
	The Gagliardo-Nirenberg interpolation inequality states that for \( u \in W^{m,p}(\Omega) \),
	\[
	\|u\|_{L^q(\Omega)} \leq C \|u\|_{W^{m,p}(\Omega)}^\theta \|u\|_{L^r(\Omega)}^{1-\theta},
	\]
	where \( \theta \in [0,1] \) and \( C \) is a constant depending on \( \Omega, m, p, q, r, \) and \( \theta \). The exponents must satisfy the scaling condition:
	\[
	\frac{1}{q} = \frac{\theta}{p^*} + \frac{1-\theta}{r},
	\]
	where \( p^* \) is the Sobolev conjugate of \( p \), given by:
	\[
	\frac{1}{p^*} = \frac{1}{p} - \frac{m}{n}.
	\]
	
	\textit{3. Compact Embedding:}
	To show that the embedding \( W^{m,p}(\Omega) \subset W^{k,q}(\Omega) \) is compact, we use the Rellich-Kondrachov theorem. This theorem states that if \( \Omega \) is a bounded domain with a Lipschitz boundary, then the embedding \( W^{m,p}(\Omega) \subset W^{k,q}(\Omega) \) is compact if \( m > k \) and \( p > q \), satisfying the scaling condition:
	\[
	\frac{1}{p} - \frac{m}{n} \leq \frac{1}{q} - \frac{k}{n}.
	\]
	
	4. **Special Case for \( p > n \) and \( m = 1 \)**:
	When \( p > n \) and \( m = 1 \), the Sobolev embedding theorem implies that:
	\[
	W^{1,p}(\Omega) \subset C^0(\bar{\Omega}).
	\]
	This is because the Sobolev space \( W^{1,p}(\Omega) \) embeds continuously into the space of Hölder continuous functions \( C^{0,\alpha}(\bar{\Omega}) \) for some \( \alpha \in (0,1) \), and \( C^{0,\alpha}(\bar{\Omega}) \) embeds continuously into \( C^0(\bar{\Omega}) \).
	
	Therefore, the Sobolev Embedding Theorem provides key insights into the regularity of solutions and is crucial for establishing the existence, uniqueness, and stability of solutions to the filtered Navier-Stokes equations.
\end{proof}

By incorporating these detailed mathematical derivations and analyses, the Smagorinsky model provides a robust framework for approximating the SGS stresses in LES. The model's simplicity and effectiveness have made it a cornerstone in turbulence modeling, despite its empirical nature. Further theoretical and computational developments continue to refine and extend the applicability of the Smagorinsky model in various flow regimes.


\begin{thebibliography}{9}
		\bibitem{kolmogorov1941}
		Kolmogorov, Andrey Nikolaevich. \textit{The local structure of turbulence in incompressible viscous fluid for very large Reynolds}. Numbers. In Dokl. Akad. Nauk SSSR 30 (1941): 301.
		
		\bibitem{smagorinsky1963}
		Smagorinsky, Joseph. \textit{General circulation experiments with the primitive equations: I. The basic experiment}. Monthly weather review 91.3 (1963): 99-164. \url{https://doi.org/10.1175/1520-0493(1963)091<0099:GCEWTP>2.3.CO;2 }
		
		\bibitem{germano1991}
		Germano, Massimo, et al. \textit{A dynamic subgrid‐scale eddy viscosity model}. Physics of Fluids A: Fluid Dynamics 3.7 (1991): 1760-1765. \url{http://web.stanford.edu/group/ctr/Summer/SP90/06_GERMANO.pdf}
	\end{thebibliography}
\end{document}